\documentclass[a4paper,11pt]{article}
\usepackage[portrait, margin=1in]{geometry}
\usepackage[utf8]{inputenc}
\usepackage[numbers, sort, compress]{natbib}
\usepackage{nicefrac}
\usepackage{multirow}

\usepackage{xspace}

\usepackage[dvipsnames,table]{xcolor}

\usepackage{float}

\usepackage{amsmath, mathtools, amsthm, amssymb}
\usepackage{algorithm, yhmath}
\usepackage[noend]{algpseudocode}

\usepackage{nicefrac}
\usepackage{thmtools}
\usepackage{thm-restate}

\usepackage{tikz}

\usetikzlibrary{calc, shapes.geometric, arrows, automata, decorations.pathreplacing, arrows.meta}

\usepackage{hyperref}
\hypersetup{
    pdftitle={Improved Bounds with a Simple Algorithm for Edge Estimation for Graphs of Unknown Size},
    pdfpagemode=FullScreen,
    }
    
\usepackage{cleveref}

\usepackage[textwidth=4em,]{todonotes}

\theoremstyle{plain}
\newtheorem{theorem}{Theorem}
\newtheorem{lemma}[theorem]{Lemma}

\theoremstyle{definition}
\newtheorem{definition}[theorem]{Definition}

\theoremstyle{remark}

\usepackage[most]{tcolorbox}
\newtcolorbox{idea}[1][]
{
colbacktitle=cyan,
colback=cyan!10,
arc=1pt,
boxrule=1pt,
title=#1 
}

\newtcolorbox{update}[1][]
{
colbacktitle=gray,
colback=gray!10,
arc=1pt,
boxrule=1pt,
title=#1 
}

\newtcolorbox{question}[1][]
{
coltitle=black,
colbacktitle=yellow,
colback=yellow!10,
arc=1pt,
boxrule=1pt,
title=#1 
}

\newtcolorbox{note}[1][]
{
coltitle=black,
colbacktitle=green,
colback=green!10,
arc=1pt,
boxrule=1pt,
title=#1 
}

\newtcolorbox{problem}[1][]
{
coltitle=black,
colbacktitle=red!60,
colback=red!10,
arc=1pt,
boxrule=1pt,
title=#1 
}

\newcommand{\bisq}{\texttt{BIS}}
\newcommand{\isq}{\texttt{IS}}
\newcommand{\iscq}{\texttt{ISC}}

\newcommand{\graph}{G}
\newcommand{\vertexset}{V}
\newcommand{\edgeset}{E}
\newcommand{\vertexcount}{n}
\newcommand{\edgecount}{m}

\newcommand{\edge}{e}

\newcommand{\samplededges}{S}

\newcommand{\approxedgecount}{\widehat{\edgecount}}

\newcommand{\collisioncount}{r}
\newcommand{\resamplesize}{s}

\newcommand{\iscollision}{\texttt{IsCollision}}
\newcommand{\collision}{k}
\newcommand{\preresamplesize}{t}
\newcommand{\birthdayestimate}{\texttt{BirthdayEstimate}}
\newcommand{\addestimation}{\texttt{ApproxSubset}}
\newcommand{\mulestimation}{\texttt{MulEstimation}}
\newcommand{\fractionofedges}{p}
\newcommand{\approxfractionofedges}{\widehat{\fractionofedges}}
\newcommand{\edgeestimationISC}{\texttt{EdgeEstimation}}
\newcommand{\threshold}{\tau}
\newcommand{\iscoll}{coll}
\newcommand{\subsequence}{\texttt{ApproxSequence}}
\newcommand{\X}{X}

\newcommand{\birthdayevent}{\mathcal{E}_\mathrm{bir}}
\newcommand{\subseqevent}{\mathcal{E}_\mathrm{sub}}

\newcommand{\constant}{c}

\newcommand{\iterations}{\kappa}
\newcommand{\thresholdconst}{\zeta_\mathrm{sp}}
\newcommand{\trueiter}{t^*}

\newcommand{\field}[1]{\mathbb{#1}}
\newcommand{\R}{\field{R}}

\newcommand{\fbrac}[1]{\left({#1}\right)}
\newcommand{\sbrac}[1]{\left\{{#1}\right\}}
\newcommand{\tbrac}[1]{\left[{#1}\right]}
\newcommand{\abs}[1]{\left|{#1}\right|}
\newcommand{\size}[1]{\left|{#1}\right|}

\DeclareMathOperator*{\E}{\field{E}}
\DeclareMathOperator*{\Var}{\mathrm{Var}}

\newcommand{\indicator}{\mathbf{1}}

\newcommand{\uniform}{\mathrm{Unif}}

\newcommand{\approxerror}{\varepsilon}
\newcommand{\confidence}{\delta}

\newcommand{\bigo}[1]{O\fbrac{{#1}}}

\title{From Decision to Random Certificates: Exponential Separation for Edge Estimation with Independent Set Queries}

\usepackage{authblk} 

\author[1]{Debarshi Chanda}
\author[1]{Buddha Dev Das}
\author[1]{Arijit Ghosh}
\author[2]{Gopinath Mishra}

\affil[1]{Indian Statistical Institute, Kolkata, India}
\affil[2]{Institute of Mathematical Sciences, HBNI, Chennai, India }

\date{}

\begin{document}
\maketitle
\thispagestyle{empty}

\begin{abstract}
We study the problem of estimating the number of edges in an undirected,
unweighted graph using sublinear query access. We consider a query model that
preserves the structure of Independent Set (IS) queries, but augments their
output with a random certificate: given a vertex subset, the oracle returns a
uniformly random edge from the induced subgraph if one exists, and returns
null otherwise.

Using this access, we give a randomized algorithm that outputs a
$(1 \pm \varepsilon)$-approximation to the number of edges with constant
success probability using $\widetilde{O}(\log^{2} m)$ queries. This implies
an exponential separation from both standard IS queries and global
random edge-sampling models: estimating the number of edges using
standard IS queries require
$\widetilde{\Theta}\!\left(\min\left\{\sqrt{m},\, \frac{n}{\sqrt{m}}\right\}\right)$
queries, while direct random edge-sample access requires
$\widetilde{\Theta}(\sqrt{m})$ samples. Beyond separation in query complexity, our algorithm is output-sensitive: its query complexity is polylogarithmic in the number of edges in the graph. This aligns with the
classical objective in group testing, where one seeks algorithms that are both
worst-case optimal and instance-adaptive.

Conceptually, our model connects
group testing, the decision--versus--counting dichotomy, graph property
testing, and the ``power of a random certificate'', and can be viewed as a
structured form of conditional sampling of edges in graphs.
\end{abstract}

\newpage


\newpage
\setcounter{page}{1}


\section{Introduction}

We consider the problem of estimating the number of edges in an undirected,
unweighted, simple graph using a sublinear number of queries. Let
$\graph = (\vertexset,\edgeset)$ be a graph with
$\vertexset = [\vertexcount] := \{1,2,\dots,\vertexcount\}$ vertices and
$\size{\edgeset} = \edgecount$ edges.
Given query access to $\graph$ and an approximation parameter
$\approxerror \in (0,1)$, the objective is to output an estimate
$\approxedgecount$ such that
\[
(1-\approxerror)\edgecount
\;\le\;
\approxedgecount
\;\le\;
(1+\approxerror)\edgecount
\]
with success probability at least $2/3$.

\paragraph{Query model and our result.}
We work in a query model built on \emph{Independent Set} (\isq{}) queries,
introduced by \citet{harpeled18}. The \isq{} query takes a subset of vertices
$U \subseteq \vertexset$ and returns whether $U$ is an independent set.

We consider the \emph{Independent Set with Certificate} query, denoted by
$\iscq{}$, which takes as input a subset of vertices $U \subseteq \vertexset$
and returns
\[
\iscq(U) =
\begin{cases}
\edge \sim \uniform(\edgeset(U)), & \text{if } \edgeset(U) \neq \emptyset,\\[4pt]
0, & \text{otherwise},
\end{cases}
\]
where $\edgeset(U)$ denotes the set of edges in the subgraph of $\graph$
induced by $U$, and $\uniform(S)$ denotes the uniform distribution over a
(nonempty) set $S$.

Importantly, the $\iscq{}$ query preserves the structural constraint of a
standard $\isq{}$ query: the algorithm specifies exactly the same vertex
subsets as in the IS model, and the only difference lies in the response,
which returns a uniformly random edge as a certificate whenever the induced
subgraph is nonempty.

\citet{harpeled18} showed that the number of edges in a graph can be estimated using
\[
\widetilde{O}\!\left( \min\!\left\{ \frac{n^{2}}{m}, \sqrt{m} \right\} \right)
\;=\;
\widetilde{O}(n^{2/3})
\]
$\isq{}$ queries, establishing a polynomial separation from local query based algorithms that require $\Theta(n)$ queries in the worst case~\citep {eden_lower_bound18}. In a  subsequent work,
\citet{Xichen19} nearly resolved the query complexity of edge estimation using
IS queries by showing that
\[
\widetilde{\Theta}\!\left( \min\!\left\{ \frac{n}{\sqrt{m}}, \sqrt{m} \right\}
\right)
\]
queries are both necessary and sufficient. In particular, this implies a
worst-case complexity of $\widetilde{\Theta}(\sqrt{n})$ $\isq{}$ queries.

Our main result shows that the natural enrichment of $\isq{}$ queries with random certificates yields an exponential separation from standard $\isq{}$ queries for edge estimation, achieving polylogarithmic query complexity in $m$.

\begin{restatable}[Main result]{theorem}{MainUB}\label{Thm: Main Upper Bound}
There exists a randomized algorithm that, given access to a graph $\graph$
through $\iscq{}$ queries and an approximation parameter
$\approxerror \in (0,1/3)$, outputs an estimate $\approxedgecount$ of the number
of edges $\edgecount$ in $\graph$ such that, with probability at least
$\nicefrac{2}{3}$,
\[
(1-\approxerror)\edgecount
\;\le\;
\approxedgecount
\;\le\;
(1+\approxerror)\edgecount .
\]
The algorithm uses
$O(\approxerror^{-2}\log^{2}\edgecount \cdot \log\log \edgecount)$
$\iscq{}$ queries in expectation.
\end{restatable}

The query complexity of the algorithm is \emph{strictly output-sensitive} in
the sense that it depends only on the number of edges, with no dependence on
$n$, even in lower-order terms.

\Cref{Thm: Main Upper Bound} not only demonstrates the power gained by
augmenting $\isq{}$ access with random certificates over standard $\isq{}$
access, but also underscores the strength of sampling within induced
subgraphs. When edges can only be sampled uniformly at random from the entire
graph, estimating the number of edges requires
$\widetilde{\Theta}(\sqrt{m})$ samples~\citep{Mitzenmacher_Upfal_2005}.  In contrast, allowing the algorithm to sample a uniformly random edge from a queried induced subgraph yields an exponential separation from global random edge-sampling models, achieving query complexity $\widetilde{O}(\log^2 m)$.


In what follows, we highlight conceptual connections between our results and
query model with several well-studied themes, including group testing, graph
property testing, the power of random certificates, decision versus counting
problems, and conditional sampling.

\paragraph{Group testing perspective.}
The introduction of $\isq{}$ queries in graph algorithms was originally
motivated by ideas from \emph{group testing}. In classical group testing,
introduced by \citet{Dorfman/AnnalsofMathStat/1943/GroupTestingFirst} and
studied extensively since then, the goal is to infer the size of an unknown
set of defective items by querying subsets of a known universe
\citep{AldridgeJS/Book/2019/GroupTestingIT}. A standard group test returns only
a binary result indicating whether the queried subset contains at least one
defective item.

This viewpoint applies naturally to graph edge estimation. Here, edges play
the role of defective items, and the allowed vertex pairs for queries impose
structural constraints on the possible subsets of edges that can be queried.
\citet{harpeled18} introduced both $\isq{}$ and Bipartite Independent Set
($\bisq{}$) queries through this lens. A $\bisq{}$ query takes two disjoint
vertex sets as input and returns whether there exists an edge with one
endpoint in each set.

In the same vein, \citet{BishnuGM23} and \citet{AssadiCK21} introduced the
Edge Emptiness query $\mathrm{EE}$ to study triangle estimation and
connectivity in graphs, respectively. An $\mathrm{EE}$ query specifies a set
of candidate edges and reports whether at least one of them exists in the
graph.

\paragraph{Decision versus counting.}
In a series of seminal works,
\citet{Stockmeyer/stoc/83/ComplxtyApprxCnt,
Stockmeyer/siamcomp/85/ApprxAlgoforSharpP} studied the relationship between
approximate counting and decision access for Boolean formulas. In particular,
they showed that given decision access to satisfiability, one can
approximately count the number of satisfying assignments of any Boolean
formula on $n$ variables using only $O(\log n)$ such queries.

Our work establishes an analogous phenomenon in the setting of graph edge
estimation: augmenting decision-style $\isq{}$ queries with a random
certificate enables efficient approximate counting of edges in a graph.

\paragraph{Conditional sampling.}
Our model also bears a close conceptual relationship to work on property
testing and estimation with conditional sampling access. In distribution
testing, a conditional sampling oracle~\citep{ChakrabortyFischerGoldhirshMatsliah/ITCS/2013/CondSamp,DBLP:journals/siamcomp/CanonneRS15/Conditional} allows the algorithm to specify a
subset of the domain and receive a sample drawn from the distribution
restricted to that subset. Such structured sampling can lead to exponential
savings over standard sampling in tasks such as uniformity testing, support
size estimation, and testing of joint distributions
\citep{CanonneRS14,BhattacharyyaC18,CanonneCKLW21,Bhattacharyya0P24}.

In our setting, the $\iscq{}$ oracle plays an analogous role: the algorithm
specifies a vertex subset, and if the induced subgraph is nonempty, it
receives a uniformly random edge from that subgraph. Thus, $\iscq{}$ can be
viewed as a structured conditional sampling primitive over the edge set of
the graph, where the conditioning is constrained to induced subgraphs. Our
results demonstrate that, as in distribution testing, such conditional
access can fundamentally change the complexity of estimation problems in
graphs.

\paragraph{Power of a random certificate.}
Studying the power of random certificates for approximate counting is an
important question across a variety of query models
\citep{RonTsur/TOCT/2016/PowerofAnExample,
ChakrabortyCKM/ICALP/2023/SATvsNPOracle}. These works show that augmenting
group queries with access to random certificates can lead to significant
reductions in query complexity.

In the graph setting, several queries have been introduced from the viewpoint
of group queries while retaining the structural properties of the graph. Two
notable examples are $\isq{}$ and $\bisq{}$, introduced by
\citet{harpeled18} for estimating the number of edges. Our $\iscq{}$ query fits
naturally within this framework: it preserves the structure of $\isq{}$ while
adding a random certificate. In particular, our algorithm achieves
substantially lower query complexity compared to algorithms using only
$\isq{}$ access, and also improves over the best known
$\widetilde{O}(\log^{5} n)$-query algorithm based on $\bisq{}$ access due to
\citet{addanki_et_al:LIPIcs.ESA.2022.2}.

The main difference between the power of random certificates studied in
\citep{RonTsur/TOCT/2016/PowerofAnExample,
ChakrabortyCKM/ICALP/2023/SATvsNPOracle} and our setting is that our query
oracles are group graph queries that retain the structural properties of the
underlying graph.

\paragraph{Graph property testing.}
Estimating the number of edges in a graph with sublinear queries has been
studied in several access models. Early works focused on local queries such as
degree, neighbour, and pair queries
\citep{feige04,avg_deg_danaron08,eden_lower_bound18,
TetekThorup/STOC/2022/EdgeSamplingFullNbrhood}, which provide fine-grained
access but face inherent limits for edge estimation.

This led to models that incorporate global information. In particular,
augmenting local access with random edge sampling
\citep{Aliakbarpour/Algorithmica/2018/StarSubgraphEdgeSampling,
EdenRS19,BerettaCS/SODA/2026/FasterEdgeEstimation} enables more efficient
estimation. More recently, a finer-grained analysis of edge estimation under
local queries enriched with random edge sampling have provided an understanding of the tradeoffs
between local and global access
\citep{BishnuCDM2025,BishnuCM_degree_distribution_2025,
Chanda_edge_estimation_unknown,BerettaCS/SODA/2026/FasterEdgeEstimation}.

Another line of work studies structured global queries, notably $\isq{}$ and
$\bisq{}$ \citep{harpeled18,Xichen19,addanki_et_al:LIPIcs.ESA.2022.2}. These
queries provide decision-only access over vertex subsets while preserving
graph structure, and have been extended to hypergraphs
\citep{DellLM22,BhattacharyaBGM24,DellLM24}. We discussed their known upper and
lower bounds earlier in the Introduction.

Finally, \citet{AdarHotamLevi/Arxiv/2601.21457} showed that augmenting
$\isq{}$ queries with additional local queries yields a quadratic separation
from using $\isq{}$ queries alone, further illustrating the power of combining
different query primitives. Other than the \bisq{} query, all query models discussed above requires polynomial number of queries in the worst case.

\subsection*{Notations}

Throughout, the notation
$\widetilde{O}(f)$, suppress factors polynomial in
$1/ \approxerror$ and $\log f$. For an integer $k$, we define $[k] = \sbrac{1,\dots,k}$.
We write $a = (1 \pm \approxerror) b$ to indicate that $a$ is a $(1 \pm \approxerror)$-approximation of $b$. 

Throughout, we use $G = (\vertexset, \edgeset) $ to denote the input graph with $ \edgecount = \abs{\edgeset} $ number of edges  and $ \vertexset = \abs{\vertexset} $ vertices. For a subset $ U \subseteq \vertexset $, let $ \edgecount( U ) $ be the number of edges in the graph induced by the subset $ U $. 





\color{black}

\section{Technical overview}

We give a high-level description of our algorithm for estimating the number of edges in an undirected graph using Independent Set with Certificate (\iscq{}) queries. Conceptually, our approach converts approximate counting into a sequence of controlled decisions and sampling steps that gradually reduce the graph while preserving a constant fraction of its edges. 

A central conceptual contribution is the use of \emph{approximation preserving subsets} and \emph{approximation preserving sequences}, which formalise how much of the graph’s edge mass is retained at each stage and guide both the algorithm and its analysis. We leverage graph sparsification techniques to enable iterative construction of approximation-preserving sequences on graphs using \iscq{} queries.

\paragraph{Outline of the Algorithm.} At a high level, our algorithm maintains a nested sequence of vertex subsets \[ \vertexset = U_0 \supseteq U_1 \supseteq \cdots \supseteq U_\kappa, \] where each $U_j$ contains a constant fraction of the edges of $U_{j-1}$. The algorithm alternates between two modes: 

\begin{itemize} 

\item If the current induced subgraph is ``small'' (that is, has fewer than $\threshold$ edges), we estimate its edge count directly. 

\item Otherwise, we construct a smaller approximation preserving subset of the graph by randomly partitioning the vertex set and selecting a part that provably retains a constant fraction of the edges. 

\end{itemize} 

This iterative reduction of the graph continues until we reach a subset $U_\kappa$ whose edge count is small enough to estimate accurately; we set the threshold for $U_\kappa$ to be a sufficiently large constant. The final estimate of $m$ is then reconstructed by multiplying the relative edge losses across all steps.

\paragraph{Where the $\widetilde{O}(\log^2 m)$ query complexity comes from.} A key novelty of our approach is that all complexity bounds are driven by $\log m$ rather than $\log n$. This is achieved through three ideas that are not present in prior $\isq{}$ or $\bisq{}$ based algorithms: 
\begin{itemize} 

\item \emph{Random certificates enable output-sensitive sparsification:} Unlike decision-only $\isq{}$ or $\bisq{}$ queries, ISC queries allow us to estimate how many edges of a subset lie in a randomly chosen part, which lets us iteratively construct smaller induced subgraphs to work on while tracking edge mass locally. 

\item \emph{Approximation Preserving Sequence:} Instead of estimating $m$ directly, we construct a chain of subsets $U_0 \supseteq U_1 \supseteq \cdots \supseteq U_\kappa$ such that each $U_j$ retains a constant fraction of edges of $U_{j-1}$.  
Crucially, we establish $\kappa$ to be $O(\log m)$ in expectation, independent of $n$. We also bound the number of queries required to construct this sequence as $O(\kappa^2)$ and bound its expectation.

\item \emph{Local multiplicative reconstruction:} We only need to estimate the ratios $p_j = \frac{m(U_j)}{m(U_{j-1})}$, each of which is $\Omega(1)$ due to the $U_0,\ldots,U_\kappa$ being an approximation preserving subsequence. This keeps the per-level query complexity small and makes the overall cost in expectation $\widetilde{O}(\kappa^2) = \widetilde{O}(\log^2 m)$. 

\end{itemize}

\paragraph{Estimating Small Edge Counts.} 

Fix a threshold $\threshold$. One can design an algorithm that makes $O(\sqrt{\tau})$ \iscq{} queries and guarantees the following with desired probability about a given subset $U \subseteq V$. If $m(U)\leq \tau$, then the algorithm reports an approximation of $m(U)$. Moreover, if $m(U) \geq 8\tau$, the algorithm reports that $m(U)$ is large. The proof is based on the birthday paradox. The details can be found in \Cref{sec: threshold checking} and Appendix \ref{appendix: threshold checking}.

\paragraph{Graph Sparsification.}

The idea here is to leverage graph sparsification to construct a smaller subset of vertices to work on at each step. We partition $U$ uniformly at random into two subsets $U^1$ and $U^2$. A graph sparsification lemma, along the result of \citet{harpeled18}, shows that with desired probability,

\[ 
    \left| \frac{\edgecount(U)}{2} - \sum_{i=1}^{2} \edgecount(U^i) \right| \leq c \sqrt{\edgecount(U)} \log \edgecount(U). 
\]

We provide the proof in \Cref{sec: sparsification}. Intuitively, this means that the total number of edges inside $U^1$ and $U^2$ is close to half of $\edgecount(U)$, up to a lower-order error term. As long as $\edgecount(U)$ is sufficiently large, this implies that at least one of $U^1$ or $U^2$ must contain a constant fraction of the edges of $U$. More precisely, there exists an absolute constant $\thresholdconst$ such that if $\edgecount(U)\ge \thresholdconst$, then with probability at least $1-1/\edgecount(U)$, at least one of $m(U^1)/m(U)$ and $m(U^2)/m(U)$ is at least $1/8$, and both $m(U^1)/m(U)$ and $m(U^2)/m(U)$ are at most $3/4$.

\paragraph{Choosing the Good Part.} After partitioning $U$ into $U^1$ and $U^2$, we must identify which part retains an appropriate fraction of the edges. Note that this step is fundamentally enabled by the $\iscq{}$ queries. The idea is to take a random sample of the edges in $U$, and check whether it is preserved in the induced subgraph of either $U^1$ or $U^2$.

To do this, we draw random edges from $U$ via \iscq{} queries. Let $\widehat{p}(U^i)$ be the fraction of sampled edges that fall inside $U^i$, for $i \in [2]$. If $\widehat{p}(U^i) \ge 1/16$, we select $U^i$ as the next subset; otherwise, we pick an arbitrary part. A Chernoff bound guarantees that, with desired probability, we correctly identify a subset satisfying 
\[ 
    \frac{1}{32}\leq \frac{\edgecount(U^i)}{\edgecount(U)} \le \frac{3}{4}. 
\] 
Such subsets are called \emph{approximation preserving subsets}. The upper bound of $3/4$ is useful to bound the final query complexity and the lower bound of $1/32$ is useful for guaranteeing the quality of the estimator.

\paragraph{Constructing an Approximation Preserving Sequence.} Starting with $U_0=\vertexset$, we iterate: 

\begin{enumerate} 

\item Run threshold checking on $U_j$. 

\item If $\edgecount(U_j) < \threshold$, stop and set $U_\kappa = U_j$. 

\item Otherwise, sparsify $U_j$ and select an approximation preserving subset $U_{j+1}$. 

\end{enumerate} 

This yields a sequence \[ \vertexset = U_0 \supseteq U_1 \supseteq \cdots \supseteq U_\kappa, \] where, with probability at least $4/5$, each $U_j$ satisfies 
\[ 
    \frac{1}{32}\leq \frac{\edgecount(U_j)}{\edgecount(U_{j-1})} \le \frac{3}{4}. 
\]

\paragraph{Final Estimation.} For each $j\in [\iterations]$, define 
\[ 
\fractionofedges_j = \frac{\edgecount(U_j)}{\edgecount(U_{j-1})}. 
\] 
Then we have the identity 
\[ 
\edgecount = \frac{\edgecount(U_\kappa)}{\prod_{j=1}^{\kappa} \fractionofedges_j}. 
\] 
Instead of estimating each $m(U_j)$ from scratch, we estimate only the ratios $p_j$. Because each $p_j = \Omega(1)$ by construction of the approximation preserving sequence, we can estimate each $p_j$ using small number of queries,  independent of $n$ or $m(U_j)$. This yields an estimate $\widehat{\fractionofedges}_j$ satisfying the following with desired probability. 
\[ 
    \widehat{\fractionofedges}_j = (1\pm \approxerror/O( \sqrt{\kappa}))\,\fractionofedges_j 
\] 

Note that the entire approximation preserving subsequence is constructed a priori. Hence, we know $\kappa$ during the final estimation phase that we use to fix the approximation parameters appropriately for the multiplicative estimation for each $p_j, j \in [\kappa]$.
Let $\approxedgecount(U_\kappa)$ be the estimate obtained via the birthday-based estimator. Our final output is \[ \widehat{\edgecount} = \frac{\approxedgecount(U_\kappa)} {\prod_{j=1}^{\kappa} \widehat{\fractionofedges}_j}. \] A careful multiplicative error analysis shows that this is a $(1\pm \approxerror)$-approximation to $m$.

\paragraph{Proof Roadmap.} The correctness proof proceeds in the following stages: 
\begin{itemize} 
\item \emph{Threshold correctness:} We show that the collision-based test reliably distinguishes between $\edgecount(U)<\threshold$ and $\edgecount(U)\ge 8\threshold$, and that the birthday-paradox estimator is accurate when $\edgecount(U)<\threshold$. 

\item \emph{Constructing approximation preserving subsequence:}  {While we decrease the failure probability of choosing a good subset at each iteration, the probability of failure of the sparsification step depends on the number of edges remaining. This necessitates a careful analysis of the failure probability of the sparsification step and of our threshold-checking mechanisms to establish the correctness of our algorithm and query complexity bounds.}

\item \emph{Multiplicative reconstruction:} Given our approximation preserving subsequence, we estimate each $p_j$ with appropriately scaled approximation error. Then, we bound the cumulative error from estimating all $\fractionofedges_j$ and show that the product estimator preserves a global $(1\pm \approxerror)$ approximation. 

\item \emph{Complexity analysis:} We show that $\mathbb{E}[\kappa] = O(\log m)$ and that the total query complexity is dominated by estimating the $p_j$'s, yielding an overall expected complexity of $\widetilde{O}(\log^2 m/\varepsilon^2)$. 

\end{itemize}

\paragraph{Outline}
We begin in \Cref{sec: threshold checking}, where we present the threshold-checking algorithm and show how to estimate the edge count when the number of edges is below a given threshold. In \Cref{sec: sparsification}, we describe the algorithm for constructing an approximation-preserving subsequence. In \Cref{section: approx subset}, we establish how graph sparsification yields an approximation-preserving subset, and use this result in \Cref{section: approximate sequence} to construct an approximation-preserving subsequence. Finally, in \Cref{section: full algorithm}, we present the complete algorithm for edge estimation.


\section{Threshold checking}\label{sec: threshold checking}

Let a fixed threshold ($ \threshold \geq 1$) be given. We consider the following threshold-checking problem: given a subset of vertices $ U \subseteq \vertexset $, determine whether $ \edgecount( U ) < \threshold $. If this condition holds, the goal is to estimate the $ (1 \pm \approxerror ) $ approximation of $ \edgecount( U )$.

To address this task, we present two algorithms \iscollision{} and \birthdayestimate{} based on the birthday paradox techniques \citep[Chapter 5]{Mitzenmacher_Upfal_2005}. \iscollision{} distinguishes between the cases $ \edgecount( U ) < \threshold $ or $ \edgecount( U ) \geq 8\threshold $ (\Cref{lemma: collision detection}). \birthdayestimate{} outputs an $ (1 \pm \approxerror) $ approximation of $ \edgecount( U ) $, if $ \edgecount( U ) < \threshold $ (\Cref{lemma: birthdayestimate}). The corresponding guarantees are stated below. For details, refer to Appendix~\ref{appendix: threshold checking}.


\begin{restatable}[Correctness of \iscollision{}]{lemma}{isollison}
    For a subset $ U \subseteq \vertexset $ and a threshold $ \threshold \geq 1 $, the following statements hold for \iscollision{} with probability at least $ 1 - \confidence$:\label{lemma: collision detection}
\begin{enumerate}
    \item[(i)] If $ \edgecount( U ) \leq  \threshold $, then $\iscollision{}$ outputs $1$.
    \item[(ii)] If $ \edgecount( U ) \geq 8 \threshold $, then $\iscollision{}$ outputs $0$.
\end{enumerate}
The number of \iscq{} queries made is $ O( \sqrt{ \threshold } \log ( \confidence^{ -1} ))$
\end{restatable}

Suppose that a collision is detected, i.e., \iscollision{} outputs $1$. We can obtain a
$ (1 \pm \approxerror)$-approximation of $\edgecount(U)$ using additional
\iscq{} queries. This estimation procedure is implemented in \birthdayestimate{}.


\begin{restatable}[Correctness of \birthdayestimate{}]{lemma}{birth}\label{lemma: birthdayestimate}
For a subset $ U \subseteq \vertexset $, threshold $ \threshold \geq 1 $ and approximation error $ \approxerror > 0 $. If $ \edgecount( U ) \leq  \threshold $, then \birthdayestimate{} outputs an $ (1 \pm \approxerror) $ estimate of $ \edgecount( U ) $ with probability at least $\nicefrac{ 15 }{ 16 }$. The number of \iscq{} query made is $ \Theta\fbrac{  \approxerror^{-1} \sqrt{\threshold} } $.
\end{restatable}

\section{Sparsification}
\label{sec: sparsification}

In this section, we will describe a procedure to generate an approximation preserving sequence (formally defined below) $U_0, \ldots, U_\kappa$ with probability at least $7/8$. Moreover, the expected query complexity of the procedure is $O(\log^2 m)$.

\begin{definition}[Approximation preserving subset]
For $U,U' \subset V$ such that  $U' \subseteq U$, we say that $U'$ is an \emph{approximation preserving subset} of $U$ if
\[
  \frac{1}{32} \le  \frac{m(U')}{m(U)}\le \frac{3}{4}.\]

\end{definition}

\begin{definition}[Approximation preserving sequence]
For a sequence $U_0, U_1, \ldots, U_\kappa \subseteq V$ with
$U_0 \supseteq U_1 \supseteq \ldots \supseteq U_\iterations$,
We say that $U_0, U_1, \ldots, U_\kappa$ is an \emph{approximation preserving sequence} if $U_{i-1}$ is an approximation preserving subset of $U_i$
for all $i \in [\kappa]$.

\end{definition}

Formally, we will prove the following lemma in this section. 
\begin{lemma}[Approximate sequence lemma] \label{lem: sequence}
There exists an algorithm that outputs an approximation preserving sequence  $ U_0, U_1, \ldots, U_\iterations $  with probability at least $ \nicefrac{4}{5}$. The expected value of $\kappa$ is $O(\log \edgecount)$, and the expected number of \iscq{} queries made is $ O( \log^2 m)$.
\end{lemma}

To prove the above lemma, we start with $U_0 = V$, and find $ U_i $ iteratively. Informally speaking, we stop with $U_\kappa$ when $m(U_\kappa)$ is \emph{small enough}. In \Cref{section: approx subset}, we discuss a procedure (\addestimation{} corresponding to \Cref{Lemma: AddEst + Sparsifier}) that generates a $U'$ which is an approximation preserving subset of $ U $, with specified probability, for any given $U$ as long as $m(U)$ is \emph{large}. In \Cref{section: approximate sequence}, we describe the procedure (\subsequence{} corresponding to \Cref{lemma: subsequence}) that uses \iscollision{} and \addestimation{} iteratively to generate an approximation preserving subsequence $ U_0, \ldots, U_\iterations $. Essentially, \iscollision{} is used to determine whether the number of edges in the current subset $U_i$ is small or large and \addestimation{} is used to generate a good subset as long as $m(U_i)$ is large.

\subsection{Procedure to find approximate preserving subset}\label{section: approx subset}
In this section, we describe the procedure \addestimation{} (see \Cref{alg: addestimation}) for finding the approximation preserving subset $ U' \subseteq U $, for any given $ U $ and prove its guarantee. 

\begin{algorithm}
    \caption{\addestimation($ U, \confidence $)}\label{alg: addestimation}
    \begin{algorithmic}[1]
        \Require \iscq{} query access to $G$ and a confidence parameter $\delta \in (0,1)$.
        \Ensure A subset $U' \subseteq U$.
        \State Uniformly and independently partition $ U $ into $ U^1 $ and $ U^2 $.\label{addestimation: line 1}
        \State Uniformly and independently sample $ O( \log ( \confidence^{-1} )) $ random edges from $ U $ and let $ \samplededges $ be the multi-set of sampled edges. \label{addestimation: line 2}
        \State Compute $ \approxfractionofedges ( U^i ) = \frac{ 1 }{ \size{ \samplededges } } \sum_{ \edge \in \samplededges} \indicator_{ \edge \in U^i } $ for $ i \in \{1, 2\} $. \label{addestimation: line 3}
        \If{There exists $ i \in \{1, 2\} $ such that $ \approxfractionofedges( U^i ) \geq \frac{ 1 }{ 16 } $} \label{addestimation: line 4}
            \State \Return $U^i$ as $U'$ \label{addestimation: line 5}
        \Else 
            \State \Return $U^1$ or $U^2$ arbitrarily as $U'$. \label{addestimation: line 7}
        \EndIf
    \end{algorithmic}
\end{algorithm}

Formally, we will prove the following.

\begin{lemma}[Approximate subset lemma] \label{Lemma: AddEst + Sparsifier}
Let $\thresholdconst$ be a suitably large constant. Consider \addestimation{} that takes $U \subseteq V$ as input and produces $U' \subseteq U$ as the output.
    \begin{itemize}
        \item[(i)] $\E[\edgecount(U')] \leq \edgecount(U)/2$;
        \item[(ii)] 

        If $m(U) \geq \thresholdconst$,  with probability $1-\delta-1/m(U)$, $U'$ is an approximation preserving subset of $U$.
    \end{itemize}
\end{lemma}

The construction of an approximation-preserving subset $ U' $ proceeds in two steps. First, we sparsify the graph by randomly partitioning the set $ U $ into $U^1$ and $U^2$ (see \Cref{lemma: graph sparsification for k = 2}). When $\edgecount( U )$ is large enough, this partition guarantees the existence of $U^i$ such that $\edgecount( U^i )$ is a constant fraction of $\edgecount( U )$ (\Cref{lemma: spasification k = 2}). In the second step, we identify such an approximation-preserving subset from $U^1$ and  $U^2$ (\Cref{lemma: addestimation}). Combining
these results yields the proof of \Cref{Lemma: AddEst + Sparsifier}.

We prove the following variant of the graph sparsification result of Beame et al.~\citep{harpeled18}.

\begin{restatable}[Graph Sparsification]{lemma}{graphsparsification}\label{lemma: graph sparsification for k = 2}
    For a given $ \confidence \in (0, 1)$, there exists an absolute constant $ \constant $ such that the following holds. Let $ \graph = ( \vertexset, \edgeset ) $ be a graph with $ \vertexcount $ vertices and $ \edgecount $ edges. Let $ U^1 $ and $ U^2 $ be a uniformly random partition of $ \vertexset $. Then,
    \begin{itemize}
        \item[(i)] $\E[ \edgecount(U^1) + \edgecount(U^2)] = \edgecount(U)/2$;
        
        \item[(ii)]    
        $\Pr\tbrac{ \abs{ \edgecount/2 - \sum_{ i = 1 }^{ 2 } \edgecount( U^i ) } \geq \constant \sqrt{ \edgecount } \log \fbrac{ \edgecount / \confidence} } \leq \confidence$.
    \end{itemize}
\end{restatable}

We will need the following lemma to prove the above graph sparsification lemma.

\begin{lemma}\label{lemma: support for sparsification}
    Let $ A $ be a set of $ \ell $ elements. For every element $ t \in A $, one assigns a color independently and uniformly from $ \sbrac{1, 2} $. Let $ C^{(i)} $ be the number of vertices in $ A $ with color $i$, for any $ i \in \sbrac{ 1, 2 }$. Then for $ \confidence \in (0, 1) $, we have
    \begin{equation*}
            \Pr \tbrac{ \abs{ C^{(1)} - C^{(2)} } \geq \sqrt{ { 6\ell\log (4\confidence^{-1}) }  }} \leq \confidence \text{ and } \E\tbrac{ \abs{ C^{(1)} - C^{(2)} } } \leq \sqrt{ \ell }
        \end{equation*}
\end{lemma}
\begin{proof}
     For $ t \in A$, let $ X_t $ be the indicator random variable that is $ 1 $ if $ t $ is assigned color $ 1 $ and $ 0 $ otherwise. Then, $ C^{(1)} = \sum_{t \in A} X_t $ and $ \E \tbrac{ C^{(1)} } = \ell/2$. Using Chernoff's inequality (Lemma~\ref{lemma: Multiplicative Chernoff Bound}) with $ \approxerror = \sqrt{ (6 / \ell) \log (4 / \confidence) }$, we have
    \begin{equation}\label{ineq: color 1}
        \Pr \tbrac{ \abs{ C^{(1)} - \frac{ \ell }{ 2 } } > \frac{\approxerror \ell }{ 2 } }  \leq 2 \exp \fbrac{ \frac{ 6 \log (4\confidence^{-1})}{ \ell } \cdot \frac{ \ell }{ 6 } } \leq \frac{ \confidence }{ 2 }
    \end{equation}
    Similarly, we can have the same bound for color $2$.
    \begin{equation}\label{ineq: color 2}
        \Pr \tbrac{ \abs{ C^{(2)} - \frac{ \ell }{ 2 } } > \frac{\approxerror \ell }{ 2 } }  \leq \frac{ \confidence }{ 2 }
    \end{equation}
    Observe that $ \abs{C^{(1)} - C^{(2)}} \leq \abs{ C^{(1)} - (\approxerror\ell/2) } + \abs{ C^{(2)} - (\approxerror\ell/2)} $. Therefore
    \begin{equation*}
        \Pr \tbrac{ \abs{C^{(1)} - C^{(2)}} > \approxerror \ell } \leq \confidence 
    \end{equation*}

    For $ t \in A$, let $ Y_t $ be a random variable that is $ 1 $ if $ t $ is assigned the color $ 1 $ and $ -1 $ otherwise. Then $ C^{(1)} - C^{(2)} = \sum_{ t \in A} Y_t = Y $. We have that $ \Pr \tbrac{ Y_t = 1 } = \Pr \tbrac{ Y_t = 1 } = 1/2$, then
    \begin{equation*}
        \E[ Y_t ] = 0 \text{ and } \E[ Y_t^2] = 1
    \end{equation*}
    By the independence of $ Y_t $, we get
    \begin{equation*}
         \E[ Y^2 ] = \sum_{ t \in A } \E[ Y_t^2 ] + \sum_{ u \neq v} \E[ Y_u Y_v ] = \ell 
    \end{equation*}
    Finally, we have $ \E[ \abs{Y} ] \leq \sqrt{ \E [ \abs{ Y }^2 ] } = \sqrt{ \E [ { Y }^2 ] } \leq \sqrt{ \ell }$. 
\end{proof}

We now prove \Cref{lemma: graph sparsification for k = 2}.

\begin{proof}[Proof of Lemma~\ref{lemma: graph sparsification for k = 2}] Let each vertex $t \in \vertexset$ be associated with a random variable
$Y_t \in \sbrac{1,2}$ such that $\Pr\tbrac{Y_t = 1} = \nicefrac{1}{2}$. For $i \in \sbrac{1,2}$, define $U^i = \sbrac{ t \in \vertexset \mid Y_t = i }$ and
\begin{equation*}
    f\fbrac{Y_1, \ldots, Y_{\vertexcount}}  = \edgecount(U^1) + \edgecount(U^2).
\end{equation*}
A fixed edge is counted by $f$ if and only if both its endpoints are assigned the same label,
which occurs with probability $\nicefrac{1}{2}$.
Therefore, $\E[f] = \edgecount / 2$.

    Consider the Doob martingale $ X_0, X_1, \ldots, X_\vertexcount $, where $ X_t = \E[f(Y_1, \ldots, Y_\vertexcount)|Y_1, \ldots, Y_t]$. For a fix values of $ Y_1, \ldots, Y_{t-1}$, we have,
    \begin{align*}
        X_{t-1} = \frac{ g(1) + g(2) }{ 2 }, \text{ where } g(i) = \E[f(Y_1, \ldots, Y_\vertexcount)| (Y_1, \ldots, Y_{t-1}) \cap (Y_t = i)]       
    \end{align*}
    Since, $X_t$ is either $g(1)$ or $g(2)$, we have $ \abs{X_t - X_{t-1}} \leq \abs{g(1) - g(2)}$. Let $ N(t) $ be the set of neighbours of the vertex $ t $ and $ \deg(t) = \abs{ N(t) }$. Let $ N_{ < t} = N(t) \cap [t-1] $ and $ N_{ > t} = N(t) \cap \sbrac{(t + 1), \ldots, \vertexcount}$. Let $C_{<t}^{(i)}$ and $C_{>t}^{(i)}$ be the number of vertices in $ N_{ < t} $ and $ N_{ > t} $ respectively. So, 
    \begin{equation*}
        \Delta_t = \abs{g(1) - g(2)} = \abs{ C_{<t}^{(1)} + \E\tbrac{ C_{>t}^{(1)} \mid Y_t = 1} - C_{<t}^{(2)} - \E\tbrac{ C_{>t}^{(2)} \mid Y_t = 2} }
    \end{equation*}
    The above equation is true because any edge with vertices in $ [t-1] $ will have the same contribution to $ g(1) $ and $ g(2) $. Similarly, any edge which have one vertex in $ [t-1] $ and another in $ \sbrac{(t + 1), \ldots, \vertexcount} $ or both of its vertices in $ \sbrac{(t + 1), \ldots, \vertexcount} $ will have the same contribution.
    Using Lemma~\ref{lemma: support for sparsification}, we have with probability at least $ 1 - \confidence'$, 
    \begin{equation*}
        \Delta_t \leq \sqrt{ 6 { \deg(t) } \log \frac{ 4 }{ \confidence' } } + \sqrt{ \deg(t) } \leq  2\constant \sqrt{ \deg(t) \log \frac{ 1 }{ \confidence' } } = \constant_t
    \end{equation*}
    Let $ \mathcal{B} $ be the event that there exits a $ t $ such that $ \Delta_t > \constant_t $. Let $ \vertexcount_a $ be the set of active vertices (vertices with degree at least $ 1 $). Therefore, using union bound $ \Pr[\mathcal{B}] \leq 4 \vertexcount_a \confidence' \leq 4 \edgecount \confidence' $.

    Let $ S = \sum_{ t = 1}^\vertexcount \constant_t^2 = \sum_{ t = 1}^\vertexcount 4 \constant^2 \deg(t) \log( 4 / \confidence') = O( \edgecount \log( 1 / \confidence' ) ) $ and $ s = \constant \sqrt{\edgecount} \log( 1 / \confidence' ) $. Using \Cref{lemma: martingle inequality}, we have
    \begin{equation*}
        \Pr \tbrac{ \abs{ f - \frac{ \edgecount }{ 2 }} > s } \leq 2 \exp \fbrac{ - \frac{ s^2 }{ 2S } } + \Pr \tbrac{ \mathcal{ B } } \leq \confidence' + 4 \edgecount \confidence'
    \end{equation*}
    Replacing $ \confidence' $ with $ \confidence / \constant \edgecount $ for a sufficiently large constant, we obtain 
    \begin{equation*}
        \Pr\tbrac{ \abs{ \frac{ \edgecount }{ 2 } - \sum_{ i = 1 }^{ 2 } \edgecount( U^i ) } \geq \constant \sqrt{ \edgecount } \log \fbrac{\frac{ \edgecount }{ \confidence }} } \leq \confidence
    \end{equation*}
    which completes the proof.
\end{proof}

\begin{restatable}[]{corollary}{}\label{lemma: spasification k = 2}
For a subset $U \subseteq \vertexset$, suppose there exists a sufficiently large
absolute constant $\thresholdconst$ such that $\edgecount(U) \ge \thresholdconst$.
Let $U^1$ and $U^2$ be a uniformly random partition of $U$. Then, with probability at least $1 - \frac{1}{\edgecount(U)}$, the following hold.
\begin{enumerate}
    \item[(i)] $ \frac{ \edgecount( U^1 ) }{ \edgecount( U ) } \geq 1/8$ or $\frac{ \edgecount( U^2 ) }{ \edgecount( U )} \geq 1/8$;
    \item [(ii)] $\frac{ \edgecount( U^1 ) }{ \edgecount( U ) }  \leq 3/4$ and $\frac{ \edgecount( U^2 ) }{ \edgecount( U )} \leq 3/4$.  
\end{enumerate}

\end{restatable}
\begin{proof}
Taking $\thresholdconst$ to be sufficiently large, we  have $c \sqrt{m(U)} \log m(U) \leq \frac{m(U)}{4}$. So, by taking $\delta=1/m(U)$ in Lemma~\ref{lemma: graph sparsification for k = 2}, we


\[
\Pr\!\tbrac{
    \frac{\edgecount(U)}{4}
    \le
    \sum_{i=1}^{2} \edgecount(U^i)
    \le
    \frac{3\,\edgecount(U)}{4}
}
\ge
1 - \frac{1}{\edgecount(U)} \, .
\]


This implies that, with probability at least $1 - 1/m(U)$,
at least one of $ \edgecount( U^1 ) / \edgecount( U ) $ and $\edgecount( U^2 ) / \edgecount( U )$ is at least $1/8$,
and both $\edgecount( U^1 ) / \edgecount( U )$ and $\edgecount( U^2 ) / \edgecount( U )$ are at most $3/4$.
\end{proof}

To prove the guarantees of \addestimation{} described in \Cref{alg: addestimation}, we need \Cref{lemma: addestimation}.

\begin{restatable}[]{lemma}{addest}\label{lemma: addestimation}
   With probability at least $ 1 - \confidence$, the following statements hold for \addestimation{} 
    \begin{enumerate}
    \item[(i)] If  $ \frac{ \edgecount( U^i ) }{ \edgecount( U ) } \geq \frac{ 1 }{ 8 } $, then $ \approxfractionofedges( U^i ) \geq \frac{ 1 }{ 16 } $.%
    \item[(ii)] If  $ \frac{ \edgecount( U^i ) }{ \edgecount( U ) } < \frac{ 1 }{ 32 } $, then $ \approxfractionofedges( U^i ) < \frac{ 1 }{ 16 } $.
    \end{enumerate}
The number of $\iscq$ queries made is $ O\fbrac{ \log ( \confidence^{-1} ) } $.
\end{restatable}

\begin{proof}
Consider a subset $ U^i \subseteq U $ for any $ i \in \{1,2\} $. 
For each sampled edge, define the indicator random variable
\[
X_j =
\begin{cases}
1 & \text{if the $ j^{\text{th}} $ sampled edge lies in $ U^i $,} \\
0 & \text{otherwise.}
\end{cases}
\]
Observe that $\E[X_j] = \frac{m(U^i)}{m(U)}$. 

Define the fraction of sampled edges that belong to $ U^i $ by
$
X = \frac{1}{\size{\samplededges}}
\sum_{j \in [\size{\samplededges}]} X_j.
$
where $ \samplededges $ is the multiset of sampled edges (Line~\ref{addestimation: line 2}) obtained via \iscq{} queries on $ U $. 
Note that $ X = \approxfractionofedges(U^i) $ and hence
\[
\E[\approxfractionofedges(U^i)]
= \frac{1}{\size{\samplededges}}
\sum_{j \in [\size{\samplededges}]}\E[X_j]
= \frac{\edgecount(U^i)}{\edgecount(U)}
= \frac{m(U^i)}{m(U)}.
\]

Applying Hoeffding's bound~(\Cref{lemma: hoeffding}), we get
\[
    \Pr\!\left( \left| \approxfractionofedges(U^i) -
    \frac{m(U^i)}{m(U)}
    \right|
    \ge \tfrac{1}{32}
    \right)
    \le
    \exp\!\left(
    -\frac{\size{S}}{2}\cdot (1/32)^2
    \right) \le \delta/2.
\]
The last inequation holds since $\size{\samplededges}=O(\log 1/\confidence)$. Using the union bound over $i\in \{1,2\}$, we get the desired result.
\end{proof}





We now give the proof of \Cref{Lemma: AddEst + Sparsifier}, using \Cref{lemma: graph sparsification for k = 2}, \Cref{lemma: spasification k = 2} and \Cref{lemma: addestimation}
\begin{proof}[Proof of \Cref{Lemma: AddEst + Sparsifier}]
In \addestimation{}, the set $U$ is partitioned uniformly at random into $U^1$ and $U^2$. 
From \Cref{lemma: graph sparsification for k = 2}(i),
\[
\E[\edgecount(U^1)+\edgecount(U^2)] = \edgecount(U)/2.
\]
Since $U'$ is either $U^1$ or $U^2$, it follows that
$\E[\edgecount(U')] \le \edgecount(U)/2$, proving part~(i).

\smallskip

We now prove part~(ii). Let $\mathcal{E}_{\text{apx}}$ denote the event that $U'$ is an approximation-preserving subset of $U$, that is, $ 1/32 \leq m(U')/m(U)\leq 3/4$.
Let $\mathcal{E}$ be the event that at least one of $\edgecount( U^1 ) / \edgecount( U )$ and
$\edgecount( U^2 ) / \edgecount( U )$ is at least $1/8$, and both $\edgecount( U^1 ) / \edgecount( U )$ and $\edgecount( U^2 ) / \edgecount( U )$ are at
most $3/4$. By \Cref{lemma: spasification k = 2},
\[
\Pr[\mathcal{E}] \ge 1 - \frac{1}{m(U)}.
\]

Let us condition on the event $\mathcal{E}$. It suffices to bound $\Pr[\mathcal{E}_{\text{apx}} \mid \mathcal{E}]$. 
 Then both $\edgecount( U^1 ) / \edgecount( U )$ and $\edgecount( U^2 ) / \edgecount( U )$ are at most $3/4$, and at least one of
$\edgecount( U^1 ) / \edgecount( U )$ and $\edgecount( U^2 ) / \edgecount( U )$ is at least $1/8$; w.l.o.g.\ assume $\edgecount( U^1 ) / \edgecount( U ) \ge 1/8$. If $\edgecount( U^2) / \edgecount( U ) \ge 1/32$, then since $U'\in\{U^1,U^2\}$ we immediately have $1/32 \le m(U')/m(U) \le 3/4$, that is,  $\Pr[\mathcal{E}_{\text{apx}}\mid\mathcal{E}]=1$. Otherwise $\edgecount( U^2) / \edgecount( U )<1/32$. By \Cref{lemma: addestimation}, with probability at least $1-\delta$, we obtain $\widehat{p}(U^1)\ge 1/16$ and $\widehat{p}(U^2)<1/16$, implying $U'=U^1$. Consequently $1/32 \le m(U')/m(U) \le 3/4$ holds with probability at least $1-\delta$ conditioned on $\mathcal{E}$. Thus in all cases $$\Pr[\mathcal{E}_{\text{apx}}\mid\mathcal{E}] \ge 1-\delta.$$

Therefore,
\[
    \Pr[\mathcal{E}_{\text{apx}}]
    \ge
    \Pr[\mathcal{E}] \cdot \Pr[\mathcal{E}_{\text{apx}}\mid\mathcal{E}]
    \ge
    1 - \frac{1}{m(U)} - \delta.
\]





\end{proof}

\subsection{Procedure to find approximate preserving sequence}\label{section: approximate sequence}
We now describe the algorithm \subsequence{} (see \Cref{alg: subsequence}), which constructs the sparsification sequence used by the main estimator.

\begin{algorithm}
    \caption{\subsequence{}($ \vertexset $)}\label{alg: subsequence}
    \begin{algorithmic}[1]
        \Require \iscq{} access to graph $G$.
        \Ensure A nested sequence $U_1, U_2, \cdots, U_\iterations$.
        \State Initialize $ U_0 \gets V $; $ \iscoll \gets 1 $; $ \iterations \gets 0 $  \label{subsequence: Line 1}
        \While{$ true $}\label{subsequence: Line 2}
            \State $\iscoll \gets \iscollision{}( U_\iterations, \thresholdconst, \frac{ 1 }{ 2^{ \iterations + 5 } } )$
            \If{$\iscoll$ equals $ 1 $}
                \State \textbf{break}
            \Else
                \State $ U_{ \iterations + 1 } \gets \addestimation( U_\iterations, \frac{ 1 }{ 2^{ \iterations + 5 } } ) $ \label{subsequence: Line 7}
                \State $ \iterations \gets \iterations + 1 $
            \EndIf
        \EndWhile
        \State \Return $U_1, U_2, \cdots, U_\iterations$
    \end{algorithmic}
\end{algorithm}



    
        
        


\begin{restatable}[Correctness of \subsequence{}]{lemma}{}\label{lemma: subsequence}
Consider \subsequence{} as described in \Cref{alg: subsequence}. With probability at least $ \nicefrac{4}{5}$, it outputs an approximation preserving sequence  $ U_1, U_2, \ldots, U_\iterations $ such that $\edgecount (U_{\kappa})\leq 8\thresholdconst$. Here $\thresholdconst$ is given large constant.  The number of $\iscq$ queries made is $ O( \kappa^2 )$.
\end{restatable}

\begin{proof}

\noindent\textbf{Correctness: }
Let us consider the event $\mathcal{E}_{\text{col}}$ that denotes
\begin{itemize}
    \item[(i)] There does not exist a $j$ such that the algorithm goes to the next iteration if $m(U_j) \leq \thresholdconst.$ 

    \item[(ii)] For any $j$ such that the algorithm goes to the next iteration if $m(U_j) \geq 8\thresholdconst.$ 
\end{itemize}

The probability of $\mathcal{E}_{\text{col}}$ is lower bounded by \iscollision{} succeeds in all iterations during the run of the algorithm. Using the union bound, we have
\[
\Pr[\mathcal{E}_{\text{col}}]
\ge 1 - \sum_{i = 1}^{\infty} \frac{1}{2^{i+5}}
\ge 1 - \frac{1}{2^5} \sum_{i = 1}^{\infty} \frac{1}{2^{i}}
\ge 1 - \frac{1}{16}
= \frac{15}{16}.
\]

Let us continue the argument under the conditional space that $\mathcal{E}_{\text{col}}$ has occurred.

We assume $m(U_0) = \edgecount$ is sufficiently large; hence, the algorithm does not terminate in the first iteration, given that $\mathcal{E}_{\mathrm{col}}$ has occurred. Let $U_0,\ldots, U_\kappa$ be the output of the algorithm, where $\kappa$ is a positive integer.
Conditioned on the event $\mathcal{E}_{\text{col}}$, we have $m(U_j) \geq \thresholdconst$ for every $j \leq \kappa - 1$. 

We prove the claim by induction on $\kappa$. For each $j \ge 0$, let $\mathcal{E}_j$ denote the event which is  that the algorithm executes iteration $j$ and either terminates with, $\edgecount(U_j) \leq 8\thresholdconst$, or it produces an approximation preserving subset $U_{j+1}$ such that $\edgecount(U_{j+1}) \leq \frac{3\edgecount(U_j)}{4}$.

Let $\mathcal{E}_{\le j}$ denote the event that $\mathcal{E}_i$ occurs for all $i \in [j]$, that is,
\[
\mathcal{E}_{\le j} \;=\; \bigcap_{i=1}^{j} \mathcal{E}_i .
\]
Equivalently, $\mathcal{E}_{\le j}$ is the event that the algorithm runs for $j$ iterations, such that
$U_0, \ldots, U_{j}$ is an approximation preserving sequence, and at iteration $j$ either it produces $U_{j+1}$, an approximation preserving subset of $U_{j}$ such that $\edgecount(U_{j+1}) \leq \frac{3\edgecount(U_j)}{4}$, or the algorithm terminates with
$\edgecount(U_j) \leq 8\thresholdconst$.

The induction hypothesis is that
\[
\Pr[\mathcal{E}_{\le j}] \;\ge\; 1 - \delta_{\leq j},
\]
where
\[
\delta_i
\;=\;
\frac{1}{\edgecount(U_i)}
+
\frac{1}{2^{i+5}}
,
\qquad
\text{and}
\qquad
\delta_{\le j}
\;=\;
\sum_{i=0}^{j} \delta_i .
\]

We first prove the base case. That is, $\Pr[\mathcal{E}_{0}]\geq 1-\delta_0$. Recall that we assume that $m(U_0) = m$ is sufficiently large. Moreover, we assume that event $\mathcal{E}_{\text{col}}$ has occurred and algorithm produced a sequence $U_0, \ldots,U_\iterations$ with $\iterations \geq 1$. So, $\mathcal{E}_0$ occurs if \addestimation{} reports an approximate preserving subset $U_{j+1}$. By  \Cref{Lemma: AddEst + Sparsifier}, this occurs with probability at least $1-\confidence_0$. 

For the inductive step, we mainly show that $\Pr[\mathcal{E}_j~|~\mathcal{E}_{\leq j-1}] \geq 1-\delta_j$ for $j\geq 1$. Let us argue that it is sufficient.
Note that
$\delta_{\le j} = \delta_{\le j-1} + \delta_j$ for $j\geq 1$ and 
\[
\Pr[\mathcal{E}_{\le j}]
=
\Pr[\mathcal{E}_{\le j-1} \cap \mathcal{E}_j]
=
\Pr[\mathcal{E}_{\le j-1}]\cdot
\Pr[\mathcal{E}_j \mid \mathcal{E}_{\le j-1}] .
\]

By the induction hypothesis, $
\Pr[\mathcal{E}_{\le j-1}] \ge 1 - \delta_{\le j-1}.$
We will show that
\begin{equation*}\label{eqn:ind}
    \Pr[\mathcal{E}_j \mid \mathcal{E}_{\le j-1}] \ge 1 - \delta_j .
\end{equation*}

Therefore,
\begin{align*}
    \Pr \left[ \mathcal{E}_{\le j} \right]
    &=
    \Pr \left[ \mathcal{E}_{\le j-1} \right] \cdot \Pr \left[ \mathcal{E}_j \mid \mathcal{E}_{\le j-1} \right] \\
    &\ge
    (1-\delta_{\le j-1})(1-\delta_j) \\
    &\ge
    1-\delta_{\le j-1}-\delta_j \\
    &=
    1-\delta_{\le j}.
\end{align*}


Consider any $j$ with $1\leq j \leq \kappa-1$. Under the conditional space $\mathcal{E}_{\text{col}}$ and along with the fact that we assume $m(U_0) = m$ is sufficiently large, $m(U_j) \geq \thresholdconst$ for each $j \leq \kappa-1$.

We consider two cases:  $m(U_j) \geq 8\thresholdconst$, and $\thresholdconst \leq m(U_j) \leq 8\thresholdconst$, and analyse them individually.

\paragraph{Case $1$: $(\edgecount(U_j) > 8\thresholdconst)$.}
Recall that we are assuming that the event $\mathcal{E}_{\text{col}}$ has occurred. So, 
 \subsequence{} does not terminate after executing \iscollision{} and execute \addestimation{}. Moreover it produces an approximate preserving subset $U_{j+1}$ with $\edgecount(U_{j+1}) \leq \frac{3\edgecount(U_j)}{4}$, that is event $\mathcal{E}_j$ occurs, when \addestimation{} reports an approximate preserving subset $U_{j+1}$ with $\edgecount(U_{j+1}) \leq \frac{3\edgecount(U_j)}{4}$. By  \Cref{Lemma: AddEst + Sparsifier} and Line 7 of \subsequence{}, this occurs with probability at least $1-\confidence_j$ where,

\[
\delta_j
=
\frac{1}{\edgecount(U_j)} + \frac{1}{2^{j+5}}
\]

 Hence, 
\[
\Pr[\mathcal{E}_j \mid \mathcal{E}_{\le j-1}] \ge 1-\delta_j^{(1)},
\]


\paragraph{Case $2$: $( \thresholdconst \leq \edgecount(U_j) \leq 8\thresholdconst )$.}
In this case, we show that
\[
\Pr\!\left[\mathcal{E}_j \mid \mathcal{E}_{\le j-1} \right] \ge 1-\delta_j
\]
regardless of whether \subsequence{} terminates after executing \iscollision{} or proceeds to execute \addestimation{}.

In the former case, the claim follows directly from the definition of $\mathcal{E}_j$ together with the fact that $\edgecount(U_j) \le 8\thresholdconst$. In the latter case, \subsequence{} produces an approximate preserving subset $U_{j+1}$ of $U_j$ provided that \addestimation{} returns an approximate preserving subset $U_{j+1}$. By \Cref{Lemma: AddEst + Sparsifier}, this occurs with probability at least $1-\delta_j$, where
\[
\delta_j
\le
\frac{1}{\edgecount(U_j)} + \frac{1}{2^{j+5}
}.
\]

 Hence,
\[
\Pr\!\left[ \mathcal{E}_j \mid \mathcal{E}_{\le j-1} \right]
\ge
1 - \delta_j .
\]

We are done with the proof of the inductive step. We summarize the above discussion as follows. $\Pr[\mathcal{E}_{\text{col}}]\geq 15/16$. When the event $\mathcal{E}_{\text{col}}$ occurs, $\Pr[\mathcal{E}_{j}~|~\mathcal{E}_{j-1}] \geq 1-\delta_j$ for all $j \geq 0$. Moreover, $\Pr[\mathcal{E}_{\leq j}] \geq 1-\delta_{\leq j}$ for all $j \geq 0$.  Recall that the algorithm outputs a sequence 
$U_0, \ldots, U_{\kappa-1}$. 
The events $\mathcal{E}_{\text{col}}$ and 
$\mathcal{E}_{\leq \iterations-1}$ occur with probability at least
\[
 \Pr[\mathcal{E}_{\text{col}}]
  - \confidence_{\leq \iterations-1}.
\]
Hence, the following facts hold with the same probability:

\begin{itemize}
    \item[1.] $U_0, \ldots, U_\kappa$ is an approximate preserving sequence, that is, $$ \frac{1}{32} \leq  \frac{m(U_{i})}{m(U_{i-1})} \leq \frac{3}{4},~\mbox{where}~i \in [\kappa].$$
    \item [2.] $\edgecount(U_j) \geq \thresholdconst$ for each $j \leq \iterations - 1$, 
and $\edgecount(U_\iterations) \leq 8\thresholdconst$.

\end{itemize}
Note that \subsequence{} produces an approximation-preserving sequence 
$U_1, U_2, \cdots, U_\iterations$ such that 
$\edgecount(U_{\kappa}) \leq 8\thresholdconst$ whenever both of the above 
conditions hold. This occurs with probability at least $1/5$, provided 
$\delta_{\iterations-1} \leq 1/10$, since 
$\Pr[\mathcal{E}_{\text{col}}] \geq 15/16$.

We complete the proof of correctness of \subsequence{} by showing that $\delta_{\iterations-1} \leq 1/10$.

\begin{align*}
    \confidence_{\leq \iterations - 1} 
    = \sum_{i=0}^{\iterations - 1} \frac{1}{\edgecount(U_j)} 
       + \frac{1}{2^{j+5}}
    \leq \sum_{i=0}^{\infty} \frac{1}{2^{j+4}} 
       + \sum_{i=0}^{\iterations - 1} \frac{1}{\edgecount(U_j)}
    \leq \frac{1}{16} 
       + \sum_{i=0}^{\iterations - 1} 
         \frac{3^{\iterations - 1 - i}}
              {4^{\iterations - 1 - i}\thresholdconst}.
\end{align*}

The last inequation holds since $\edgecount(U_{\iterations - 1}) \geq \thresholdconst$ and 
$\edgecount(U_i) \leq \frac{3\edgecount(U_{i-1})}{4}$ for all 
$i \in [\iterations - 1]$. We obtain the bound of $\confidence_{\leq \iterations - 1}$ as below:
\begin{align*}
    \confidence_{\leq \iterations - 1}
    = \frac{1}{16} 
       + \frac{1}{\thresholdconst}
         \sum_{i=0}^{\iterations - 1} 
         \frac{3^{\iterations - 1 - i}}
              {4^{\iterations - 1 - i}}
    = \frac{1}{16} 
       + \frac{1}{\thresholdconst}
         \sum_{i=0}^{\iterations - 1} 
         \frac{3^{i}}{4^{i}}
    \leq \frac{1}{16} + \frac{2}{\thresholdconst}\leq \frac{1}{10}.
\end{align*}
 Here we have used the fact that $\thresholdconst$ is a large constant.

    \paragraph{Query Complexity:} From Lemma~\ref{lemma: collision detection} and~\ref{lemma: addestimation}, the number of \iscq{} queries made by \iscollision{} (for a constant threshold $ \threshold = \thresholdconst $) and \addestimation{} over $\iterations$ iterations is,
    \begin{align*}
        O\fbrac{\sum_{ j = 1 }^{\iterations} \log( 2^{ \iterations + 5 } ) } = O\fbrac{\iterations^2}
    \end{align*}
\end{proof}

Note that the number of iterations and the query complexity of  \subsequence{} (as mentioned in the above lemma) are random variables. We will argue their expectation in the following lemma. Moreover, 
when we discuss the full algorithm in \Cref{section: full algorithm} (that uses \subsequence{} as a subroutine), we show that the query complexity of the full algorithm is $O(\kappa^2 \cdot \log \kappa)$, where $\kappa$ is the number of iterations of \subsequence{}. The following lemmas will be useful in establishing the final query complexity in \Cref{section: full algorithm}.

\begin{lemma}\label{lemma: kappa}
$\E[\iterations]=O(\log \edgecount)$, $\E[\iterations^2]=O(\log^2 \edgecount)$, and $\E[\iterations^2 \log \iterations]=O(\log^2 \edgecount \cdot \log \log \edgecount)$.
\end{lemma}

From \Cref{lemma: subsequence} and \Cref{lemma: kappa}, \Cref{lem: sequence} follows. In the rest of the section, we prove \Cref{lemma: subsequence} and \Cref{lem: sequence}.

\begin{proof} Here, we prove that 
$$\E\tbrac{ \iterations^2 \log \iterations} = O( \log^2 \edgecount \log \log \edgecount ) $$ The remaining bounds follow similarly.

    We want to analyse $\E[\iterations^2 \log \iterations]$ where the expectation is over the distribution of $\iterations$ for a run of the algorithm. To analyse this, we consider two different cases. For the first case, let $\trueiter$ denote the time when the algorithm generates a subgraph with less than or equal to $\thresholdconst$ edges. Formally, we define $\trueiter$ as the time $t$ such that $\edgecount(U_{t}) \leq \thresholdconst$, and $\edgecount(U_{t-1}) > \thresholdconst$. If this is not the case, i.e. when $m(U_t) > \thresholdconst$ holds throughout the algorithm, we show that $ \iterations = O( \log \edgecount )$.

    First, we bound the expectation of $\edgecount(U_t)$ for some $t \in [\iterations]$. By \Cref{Lemma: AddEst + Sparsifier}, we have
    \begin{align*}
        \E\left[\edgecount(U_t)~|~m(U_{t-1})\right] \leq \frac{\edgecount(U_{t-1})}{2}.
    \end{align*}
    By a inductive argument, and the fact that $\edgecount(U_0) = \edgecount$, we have
    \begin{align}
        \E[\edgecount(U_t)]=\E[\E[\edgecount(U_{t})~\mid~\edgecount(U_{t-1})]] \leq \frac{\E[\edgecount(U_{t-1})]}{2} \leq \ldots \leq \frac{\edgecount}{2^t} \label{Eq: Edge bound at step t}
    \end{align}
    Now, when $m(U_j) > \thresholdconst$ holds throughout the algorithm, we have
    \begin{equation*}
        \thresholdconst < \E\tbrac{ \edgecount( U_\iterations ) } \leq \frac{ \edgecount}{ 2^{\iterations} } \implies \iterations = O ( \log \edgecount )
    \end{equation*}
    
    Hence, we now focus on the case when there exists a $\trueiter \leq \iterations$. Note that $\trueiter$ is unique for a run of the algorithm. Once, the algorithm reaches $\trueiter$, the probability of the algorithm terminating depends on the calls to \iscollision{} succeeding. We start with analysing the conditional expectation of $\iterations^2 \log \iterations$ given $\trueiter$.
    \begin{align}
        \nonumber\E \tbrac{ \iterations^2  \log \iterations \mid \trueiter } &= \sum_{ j = 1 }^\infty j^2 \log j \Pr \tbrac{ \iterations = j \mid \trueiter}\\\nonumber
        &= \sum_{ j = 1}^{\trueiter -1} j^2 \log j \Pr \tbrac{ \iterations = j \mid \trueiter} + \sum_{ j = \trueiter}^\infty j^2\log j \Pr \tbrac{ \iterations = j \mid \trueiter} \\\nonumber
        &= \sum_{ j = \trueiter}^\infty j^2 \log j \Pr \tbrac{ \iterations = j \mid \trueiter} &\text{(As $\trueiter \leq \iterations$)}\\\nonumber
        & \leq\sum_{ j = \trueiter}^\infty j^2 \log j \Pr \tbrac{ \iterations \geq j \mid \trueiter}\\\nonumber
        & \leq \sum_{ j = \trueiter}^\infty \frac{ j^2 \log j }{2^{j + 5}}\\
        & = \bigo{{\trueiter}^2 \log \trueiter}\label{Eq: Conditional iteration bound} 
    \end{align}

    Now, we analyse the distribution of $\trueiter$. We use~\eqref{Eq: Edge bound at step t} to obtain a bound on the probabilities of different values of $\trueiter$. 
    \begin{align}
        \nonumber\Pr \tbrac{ \trueiter = t } \leq \Pr \tbrac{ \trueiter \geq t } &\leq \Pr \tbrac{ \edgecount( U_{ t } ) > \thresholdconst }\\\nonumber
        &\leq \frac{ \E[ \edgecount( U_{ t } ) ] }{ \thresholdconst } &\text{Markov's Inequality (Lemma~\ref{Markov's Inequality})}\\
        &\leq \frac{ \edgecount }{2^t \thresholdconst  } &\text{\Cref{Eq: Edge bound at step t}}\label{Eq: Markov Edge bound}
    \end{align}
    Now, we bound the expectation of $\iterations^2 \log \iterations$.
    \begin{align*}
        \E[ \iterations^2 \log \iterations ] &= \E \tbrac{\E \tbrac{ \iterations^2 \log \iterations \mid \trueiter } }\\
        &= \E[ {\trueiter}^2 \log \trueiter ] &\text{(\Cref{Eq: Conditional iteration bound})}\\
        &= \sum_{j = 1}^\infty j^2 \log j \Pr[\trueiter = j]\\
        &= \sum_{j = 1}^{2\log \edgecount} j^2 \log j \Pr[\trueiter = j] + \sum_{j = 2\log\edgecount + 1}^\infty j^2 \log j \Pr[\trueiter = j]\\
        &\leq 8\log^2\edgecount \log \log\edgecount\sum_{j = 1}^{2\log \edgecount} \Pr[\trueiter = j] + \sum_{2\log\edgecount}^\infty \frac{\edgecount j^2 \log j}{2^j\thresholdconst} \\
        &= \bigo{\log^2 \edgecount \log \log \edgecount} 
    \end{align*}
    which proves the desired result.
\end{proof}

\Cref{lem: sequence} follows  from the algorithm \subsequence{} (\Cref{alg: subsequence}) along with \Cref{lemma: subsequence} and \Cref{lemma: kappa}.

\section{Full algorithm}\label{section: full algorithm}

In this section, we present the full algorithm for estimating the number of edges using \iscq{} queries and prove our main result stated below.

\MainUB*

\edgeestimationISC{} (\Cref{alg: edgeestimationISC}) begins by calling \subsequence{} which outputs an approximation preserving sequence $U_1 , \ldots, U_\iterations $, where each $ \fractionofedges_i = \edgecount( U_i ) / \edgecount( U_{i-1} ) $ is bounded below by a constant. This ensures that each ratio can be estimated efficiently. To estimate each $ \fractionofedges_i $, we use the subroutine \mulestimation{} which returns an $ (1 \pm \approxerror)$-approximation of $ \fractionofedges_i$ . The procedure \mulestimation{} is described in \Cref{alg: mulestimation} and its guarantee is established in \Cref{lemma: mulestimation}.

\begin{algorithm}[H]
    \caption{\mulestimation($ U', U, \approxerror, \confidence $)}\label{alg: mulestimation}
    \begin{algorithmic}[1]
        \Require $ U' \subseteq U$ and \iscq{} access to the subset $ U \subseteq \vertexset $.
        \Ensure $\approxfractionofedges ( U' )$
        \State Uniformly and Independently sample $ O\fbrac{ \log ( \confidence^{-1} ) / \approxerror^2 } $ random edges from $ U $ and let $ \samplededges $ be the multi-set of sampled edges.
        \State \Return $ \approxfractionofedges ( U' ) = \frac{ 1 }{ \size{ \samplededges } } \sum_{ \edge \in \samplededges} \indicator_{ \edge \in U' } $. 
    \end{algorithmic}
\end{algorithm}







\begin{algorithm}[H]
    \caption{\edgeestimationISC($ \vertexset, \approxerror $)}\label{alg: edgeestimationISC}
    \begin{algorithmic}[1]
        \Require \iscq{} access to any subset $ U \subseteq \vertexset $.
        \Ensure An estimate of $ \approxedgecount $
        \State $ U_1, U_2, \cdots, U_\iterations \gets \subsequence{}(\vertexset)$
        \State $ \approxedgecount_b \gets \birthdayestimate( U_\iterations, 8\thresholdconst, \frac{\approxerror }{ 3 } )$\label{edgeestimation: line 2}
        \State Initialize $ \approxfractionofedges_i \gets 1 $, for all $ i \in [\iterations]$
        \For{$ i = 1 $ to $ \iterations $}
            \State  $ \approxfractionofedges_i \gets \mulestimation( U_i, U_{i - 1}, \frac{ \approxerror }{ 200 \sqrt{\iterations} }, \frac{ 1 }{ 16 \iterations })$ \label{edgeestimation: line 5}
        \EndFor
        \State $ \approxedgecount \gets \approxedgecount_b / { \prod_{ i = 1 }^{ \iterations } \approxfractionofedges_i }$
        \State \Return $ \approxedgecount $
    \end{algorithmic}
\end{algorithm}


\begin{restatable}[Correctness of \mulestimation{}]{lemma}{mulest}\label{lemma: mulestimation}
    For subsets $ U' \subseteq U \subseteq \vertexset $ and $ \approxerror > 0$ , If $ \fractionofedges( U' ) = \frac{ \edgecount( U' ) }{ \edgecount( U ) } \geq \frac{ 1 }{ 32 } $, then \mulestimation{} outputs $ \approxfractionofedges ( U' ) = ( 1 + x ) \fractionofedges( U' ) $, where $\E[x] = 0 $ and $x \in (-\approxerror, \approxerror)$ with probability at least $ 1 - \confidence $. The number of \iscq{} queries made is $ O\fbrac{ \log ( \confidence^{-1} )  / \approxerror^2 } $.
\end{restatable}

\begin{proof}
    Let $ X_j$ be the indicator random variable such that,
\begin{equation*}
    X_j = 
    \begin{cases}
        1 & \text{If the $ j^{th} $ sampled edge is in $ U' $} \\
        0 & \text{Otherwise }
    \end{cases}
\end{equation*}
Define number of edges sampled that belong to $ U' $ is $ X  = \sum_{ j \in [\size{ \samplededges }] } X_j$, where $ \samplededges $ is the multi-set of sampled edges (Line~\ref{addestimation: line 2}) and the expected size of $ X $ is
\begin{equation*}
    \E\tbrac{ X } =  \sum_{ j \in [\size{S}] } \E\tbrac{ X_j } = \size{ S } \cdot \frac{ \edgecount{ (U') }}{ \edgecount{ (U) } } = \size{ S } \cdot \fractionofedges( U' ) \geq \frac{ \size{ \samplededges } }{ 32 }
\end{equation*}
Therefore,
\begin{align*}
    \E[x] = \E\tbrac{ \frac{\approxfractionofedges ( U' )}{\fractionofedges( U' )}} -1 = \frac{\E[X]}{\size{ S } \fractionofedges( U' )} -1 = 0
\end{align*}
By the Chernoff Bound (see Lemma~\ref{lemma: Multiplicative Chernoff Bound}), we have 
\[
    \Pr[x \in (-\approxerror, \approxerror)] = \Pr\tbrac{ \approxfractionofedges(U') \in (1\pm\approxerror)\fractionofedges(U') } = \Pr\tbrac{ X \in (1\pm\approxerror)\E[X] } \ge 1 - 2\exp\fbrac{ - \frac{\approxerror^2 \size{\samplededges}}{96} } \ge 1 - \confidence,
\]
where the last step uses $\size{\samplededges}=O\left(\log (\delta^{-1})/\varepsilon^2\right)$.
\end{proof}

Now, we prove the guarantees of \edgeestimationISC{}
In particular, we prove the correctness of \edgeestimationISC{} and its query complexity in \Cref{lem:final-correct} and \Cref{lem:final:query}, respectively.
\begin{lemma}\label{lem:final:query}
    The expected query complexity of \edgeestimationISC{} is $O(\log^2 m \cdot \log \log m)$.
\end{lemma}

\begin{proof}
 From Lemma~\ref{lemma: mulestimation}, the number of \iscq{} queries made by \mulestimation{} in $j$-th iteration is
    \begin{align*}
        O\fbrac{ \frac{\log(16 \iterations )}{ \approxerror'^2 } } \text{where } \approxerror' = \frac{ \approxerror }{ 200 \sqrt{\iterations} }
    \end{align*}
    Summing over $ \iterations $ iterations, we get
    \begin{align*}
        O\fbrac{ \sum_{ j = 1 }^{ \iterations } \frac{ \log( \iterations ) }{ \approxerror'^2 } } = O\fbrac{\iterations^2 \frac{ \log \iterations }{ \approxerror^2 } }
    \end{align*}
    From Lemma~\ref{lemma: subsequence}, the number of \iscq{} queries made by \subsequence{} is $ O( \iterations^2 ) $. Finally, from Lemma~\ref{lemma: birthdayestimate}, the number of \iscq{} queries made by \birthdayestimate{} is $O\fbrac{ { \sqrt{8\thresholdconst} }/{ \approxerror}}= O(1/\approxerror)$. Combining all contributions, the total number of \iscq{} queries made by \edgeestimationISC{} is dominated by \mulestimation{}, and hence is, $ O( \iterations^2 \log \iterations / \approxerror^2)$. From \Cref{lemma: kappa}, 
    $$
        \E[\iterations^2 \log \iterations] = O(\log^2 \edgecount \log \log \edgecount)
    $$ 
    which bounds the expected query complexity.
\end{proof}
\begin{lemma}\label{lem:final-correct}
In \edgeestimationISC{}, the estimate $\widehat{m}$ is an
$(1 \pm \approxerror)$-approximation to $m$ with probability at least $3/4$.
\end{lemma}

    \begin{proof} 
    Let $\mathcal{E}_1 = \subseqevent \cap \birthdayevent$, where $\subseqevent$ be the event that \subsequence{} outputs an approximation-preserving sequence $U_1, U_2, \ldots, U_\iterations$ with $ \edgecount( U_\iterations ) \leq 8 \thresholdconst$. From \Cref{lemma: subsequence} we have $\Pr[\subseqevent] \geq 4/5$.
    
    $ \birthdayevent $ be the event that \birthdayestimate{} outputs an $(1 \pm \approxerror/3)$ approximation of $ \edgecount( U_\iterations ) $ from \Cref{lemma: birthdayestimate}. Also, from \Cref{lemma: birthdayestimate}, we have $\Pr[\birthdayevent \mid \subseqevent] \geq 15/16$. Therefore,
    \begin{equation*}
        \Pr[\mathcal{E}_1] = \Pr[\subseqevent]\Pr[\birthdayevent \mid \subseqevent] \geq 1 - \fbrac{ \frac{1}{5} + \frac{1}{16} } \geq \frac{59}{80}
    \end{equation*}

    Let $\mathcal{E}_2$ be the event that \mulestimation$\left( U_i, U_{i - 1}, \varepsilon', { 1 }/{ 16 \iterations}\right)$ outputs  $\widehat{p_i}=(1 \pm x_i/4)p_i$ such that $\E[x_i]=0$ and $x_i \in \fbrac{ -4{\approxerror'}, 4\approxerror'}$ for all $i \in [\iterations]$, where $ \approxerror' = \frac{ \approxerror }{ 200 \sqrt{ \iterations }}$ . Therefore, using Lemma~\ref{lemma: mulestimation} and applying the union bound over all $i \in [\iterations]$, 
    \begin{equation*}
        \Pr[\mathcal{E}_2 \mid \mathcal{E}_1] = 1 - \sum_{ j = 1 }^{ \iterations } \frac{ 1 }{ 16 \iterations } \geq \frac{15}{16}.
    \end{equation*}
    %

    Let $\mathcal{E}_3$ be the event that $\abs{\sum_{i = 1 }^\iterations x_i} \leq \approxerror/9$. 
    By Hoeffding's bound (Lemma~\ref{lemma: hoeffding}),
\begin{align}
\Pr\!\left[
\mathcal{E}_3 \mid \mathcal{E}_1 \cap \mathcal{E}_2
\right]
&=
1-\Pr\!\left[
\left|
\sum_{i=1}^{\iterations} x_i
\right|
>
\frac{\approxerror}{9}
\,\middle|\,
\mathcal{E}_1 \cap \mathcal{E}_2
\right]
\nonumber \\
&\ge 1-
2\exp\!\left(
-
\frac{
2\approxerror^2
}{
81 \sum_{i=1}^{\iterations}
\left(
\frac{2\approxerror}{50\sqrt{\iterations}}
\right)^2
}
\right)
\ge 1-\frac{1}{e^{12}}. \label{ineq: hoeff}
\end{align}

   Hence,
     \begin{align*}
        \Pr[ \mathcal{E}_1 \cap \mathcal{E}_2 \cap \mathcal{E}_3] &= \Pr[\mathcal{E}_1] \cdot \Pr[\mathcal{E}_2 | \mathcal{E}_1] \cdot \Pr[\mathcal{E}_3 \mid \mathcal{E}_1 \cap \mathcal{E}_2] \geq 1 - \fbrac{ \frac{21}{80} + \frac{ 1 }{ 16 } + \frac{ 1 }{ e^{12} } } \geq \frac{2}{3}
    \end{align*}
In what follows, we assume that all of the events $\mathcal{E}_1$, $\mathcal{E}_2$ and $\mathcal{E}_3$ occur. Since, $\Pr[ \mathcal{E}_1 \cap \mathcal{E}_2 \cap \mathcal{E}_3] \geq 2/3$,  it is enough to argue that \edgeestimationISC{} returns an $(1 \pm \approxerror)$-approximation to $m$ when all $\mathcal{E}_1$, $\mathcal{E}_2$ and $\mathcal{E}_3$ holds.

    Combining the output of \mulestimation{} over the $ \iterations $ iterations, we get,
    \begin{align*} 
          \prod_{i = 1}^{ \iterations } \approxfractionofedges_i = \prod_{i = 1 }^{ \iterations } \fractionofedges_i \cdot \prod_{i = 1}^{ \iterations } \fbrac{ 1 + \frac{x_i}{4} }
    \end{align*}
 Using the fact that for each $x_i \in (0,1)$, it can be shown that 
    \begin{align*}
        1 - \frac{ \sum_{i = 1}^{ \iterations } x_i }{ 4 }  \leq \prod_{i = 1}^{ \iterations } \fbrac{ 1 + \frac{x_i}{4} } \leq 1 + \sum_{i = 1}^{ \iterations } x_i
    \end{align*}

\noindent Since event $\mathcal{E}_3$ occurs,
    \begin{align*}
        \prod_{i = 1 }^{ \iterations } \fractionofedges_i \cdot \fbrac{1 - \frac{\approxerror}{36}} \leq \prod_{i = 1}^{ \iterations } \approxfractionofedges_i \leq \prod_{i = 1 }^{ \iterations } \fractionofedges_i \cdot \fbrac{1 + \frac{\approxerror}{9}}
    \end{align*}
    Using the fact that for any $\approxerror \in (0, 1/3)$, we have $1 - 3\approxerror \leq (1 + \approxerror)^{-1}$ and $ (1 - \approxerror)^{-1} \leq 1 + 3\approxerror $. We get, 
    \begin{align}\label{ineq: xyz}
        \fbrac{1 - \frac{\approxerror}{3}} \cdot \frac{1}{\prod_{i = 1 }^{ \iterations } \fractionofedges_i}
         \leq \frac{1}{\prod_{i = 1}^{ \iterations } \approxfractionofedges_i} \leq \frac{1}{\prod_{i = 1 }^{ \iterations } \fractionofedges_i} \cdot \fbrac{1 + \frac{\approxerror}{12}}
    \end{align}

    Since $\mathcal{E}_1$ occurs, i.e., $\birthdayevent$ occurs, $ \approxedgecount_b$  satisfies the following:
    \begin{align} \label{ineq: birthdayest}
         \fbrac{ 1 - \frac{\approxerror}{3} } \cdot \edgecount( U_\iterations )  \leq \approxedgecount_b \leq \fbrac{ 1 + \frac{\approxerror}{3} } \cdot \edgecount( U_\iterations )
    \end{align}
    Combining~\eqref{ineq: xyz} and \eqref{ineq: birthdayest}, we get,
    \begin{align*}
        \fbrac{ 1 - \frac{\approxerror}{3} }^2 \cdot \frac{ \edgecount( U_\iterations ) }{ \prod_{i = 1 }^{ \iterations } \fractionofedges_i} 
        \leq 
        \widehat{m}=\frac{ \approxedgecount_b }{ \prod_{i = 1}^{ \iterations } \approxfractionofedges_i} 
        \leq
        \fbrac{ 1 + \frac{\approxerror}{12}}\fbrac{ 1 + \frac{\approxerror}{3}} \cdot \frac{ \edgecount( U_\iterations ) }{ \prod_{i = 1}^{ \iterations } \fractionofedges_i}
    \end{align*} 
    
    Since, $ \fractionofedges_i = \frac{ \edgecount( U_i) }{ \edgecount( U_{i - 1} )}$ for each $i \in [\kappa]$ and $ U_0 = \vertexset $, $ \frac{ \edgecount( U_\iterations ) }{ \prod_{i = 1 }^{ \iterations } \fractionofedges_i} = \edgecount( U_0 ) = \edgecount $. So, simplifying the above expression, we have 
    \begin{align*}
        \fbrac{ 1 - \approxerror } \edgecount \leq \approxedgecount \leq \fbrac{ 1 + \approxerror} \edgecount 
    \end{align*}
\end{proof}


\paragraph{Acknowledgments:}{We thank an anonymous reviewer for their suggestions in improving the query complexity from $\log^3 \edgecount$ (in an earlier version) to $\log^2 \edgecount$.
 Arijit Ghosh acknowledges partial support from the Science and Engineering Research Board (SERB), Government of India, through the MATRICS grant MTR/2023/001527, and from the Department of Science and Technology (DST), Government of India, through grant TPN-104427.}

\bibliographystyle{plainnat}
\bibliography{refs}

\appendix

\section{Concentration Inequalities}

We use the following concentration bounds.

\begin{lemma}[Markkov's inequlaity, see~\citep{Mitzenmacher_Upfal_2005}]\label{Markov's Inequality}
Let $X$ be a random variable that assumes only nonnegative values. Then, for all $ a > 0 $
\begin{equation*}
    \Pr \tbrac{ X \geq a} \leq \frac{ \E [ X ] }{ a }.
\end{equation*}    
\end{lemma}

\begin{lemma}[Chebyshev’s inequality, see~\citep{Mitzenmacher_Upfal_2005}]\label{lemma: chebyshev}
    Let $ X $ be a random variable and $ a > 0 $, Then
    \begin{equation*}
        \Pr[ \abs{X - \E[X]} > a] \leq \frac{ \Var( X ) }{ a^2 }.
    \end{equation*}
\end{lemma}

\begin{lemma}[Multiplicative Chernoff bound, see~\citep{Mitzenmacher_Upfal_2005}]\label{lemma: Multiplicative Chernoff Bound}
    Let $X_1,X_2,...,X_t$ be i.i.d. random variables
    where $\Pr[X_i = 1] = p$ and $\Pr[X_i = 0] = 1-p$ for all $i \in [t]$, and $X = \sum_{i \in [t]} X_i$. Then, we have
    \begin{align*}
    \Pr \left[ X \leq (1-\approxerror) \E\tbrac{X} \right] &\leq \exp{\fbrac{-\frac{\approxerror^2\E\tbrac{X}}{3}}} & 0 \leq \approxerror <1\\
    \Pr \left[ X \geq (1+\approxerror) \E\tbrac{X} \right] &\leq \exp{\fbrac{-\frac{\approxerror^2\E\tbrac{X}}{2}}} & 
    \approxerror \geq 0
    \end{align*}
\end{lemma}

\begin{lemma}[Hoeffding's bound, see~\citep{Mitzenmacher_Upfal_2005}]\label{lemma: hoeffding}
    Let $\X_1, \X_2, \ldots, \X_t$ be independent random variables such that for all $i \in [t]$, $\E[\X] = \mu $ and $\Pr[a \leq \X_i \leq b] = 1$. Then
    \begin{equation*}
        \Pr \tbrac{ \abs{ \frac{ 1 }{ t } \sum_{i = 1}^t \X_i - \mu }  \geq \approxerror } \leq 2 \exp{ \fbrac{ - \frac{ 2t\approxerror^2 }{ (b - a)^2 } } }.
    \end{equation*}
\end{lemma}

\begin{lemma}[See~\citep{martingle06}]\label{lemma: martingle inequality}
    Let f be any function of $t$ independent random variables $Y_1, \ldots, Y_t$ and let $X_i = \E[f(Y_1, \ldots, Y_t)|Y_1, \ldots, Y_i]$ for $i \in [t]$ and $ X_0 = \E[f(Y_1, \ldots, Y_t)]$. Let $ \mathcal{B} $ be the event that there exists $ i \in [t]$ such that $ \abs{X_i - X_{i-1}} > c_i $ where $c_1, \ldots, c_t$ are some non negative numbers. Let $S = \sum_{i = 1}^t c_i^2$, then
    \begin{equation*}
        \Pr \tbrac{\abs{X_t - X_0} > s } \leq 2\exp{ \fbrac{-\frac{s^2}{2S}}} + \Pr[ \mathcal{B}].
    \end{equation*}
\end{lemma}

\section{Threshold checking}\label{appendix: threshold checking}

In this section, we present the algorithms \iscollision{}, \birthdayestimate{} and the proofs of \Cref{lemma: collision detection} and \Cref{lemma: birthdayestimate}.

\begin{algorithm}
    \caption{\iscollision($U, \threshold, \confidence$)}\label{alg:collision}
    \begin{algorithmic}[1]
        \Require \iscq{} access to the subset $U \subseteq \vertexset$ and a threshold $\threshold$
        \Ensure $0$ or $1$
        \For{$i \gets 1$ to $O(\log(\confidence^{-1}))$} \label{alg: collision line 1}
            \State  Uniformly and independently sample $\preresamplesize = \sqrt{2\threshold}$ random edges from $U$
                \label{alg: collision line 2}
        \If{there is a collision}
            \State $\collision_i \gets 1$
        \Else 
            \State $\collision_i \gets 0$
        \EndIf
        \EndFor
        \If{more than half of $\collision_i$ are $1$}
            \State \Return  $1$
        \Else 
            \State \Return $0$
        \EndIf
    \end{algorithmic}
\end{algorithm}







\isollison*

\begin{proof}
Fix an iteration $i$ of \iscollision{} (Line~\ref{alg: collision line 1}).  
The probability that no collision occurs in the $i$-th iteration is
\begin{align*}
    \Pr[\collision_i = 0]
    = \prod_{j=0}^{\preresamplesize - 1}
      \left(1 - \frac{j}{\edgecount(U)}\right).
\end{align*}

\noindent If $\edgecount(U) \leq \threshold$,
We upper bound the probability of no collision as follows:
\begin{align*}
    \Pr[\collision_i = 0]
    &\leq \exp\left(
        - \sum_{j=0}^{\preresamplesize - 1}
        \frac{j}{\edgecount(U)}
    \right)
    && \text{($1 + x \leq e^{x}$ for all $x \in \R$)} \\
    &= \exp\left(
        - \frac{\preresamplesize(\preresamplesize - 1)}{2 \edgecount(U)}
    \right) \\
    &\leq \exp\left(
        - \frac{\preresamplesize(\preresamplesize - 1)}{2 \threshold}
    \right)
    && \text{(since $\edgecount(U) \leq \threshold$)} \\
    &\leq \frac{1}{e}
    && \text{(using $\preresamplesize = \sqrt{2\threshold}$)}.
\end{align*}

\noindent If $\edgecount(U) \geq 8\threshold$,  we lower bound the probability of no collision:
\begin{align*}
    \Pr[\collision_i = 0]
    &\geq \exp\left(
        - \sum_{j=0}^{\preresamplesize - 1}
        \left(
            \frac{j}{\edgecount(U)} +
            \frac{j^2}{\edgecount(U)^2}
        \right)
    \right)
    && \text{($1 - x \geq e^{-x - x^2}$ for $x \in [0, 2/3]$)} \\
    &\geq \exp\left(
        - \frac{\preresamplesize(\preresamplesize - 1)}{16 \threshold}
        + \Theta\!\left(\frac{1}{\sqrt{\threshold}}\right)
    \right) \\
    &\geq \frac{1}{2}
    && \text{(using $\preresamplesize = \sqrt{2\threshold}$)}.
\end{align*}

\noindent If $ \edgecount( U ) \leq \threshold $, $ \Pr[ \collision_i = 0 ] \leq \frac{ 1 }{ e } $. \iscollision{} outputs $0$ when more than half of $ \collision_i = 0 $. Over $ \log ( \confidence^{-1} ) $ parallel runs, using the upper tail estimate of the Chernoff Bound (see Lemma~\ref{lemma: Multiplicative Chernoff Bound}), the probability that \iscollision{} outputs $0$ is at most $ \exp ( - \log ( \confidence^{-1} ) ) = \confidence $. Similarly, if $ \edgecount( U ) \geq 8\threshold $, the probability that \iscollision{} outputs $1$ is at most $ \confidence $.

\end{proof}

\begin{algorithm}
    \caption{\birthdayestimate{}($ U, \threshold, \approxerror $)}\label{alg: birthdayestimate}
    \begin{algorithmic}[1]
        \Require \iscq{} access to the subset $ U \subseteq \vertexset $ and a threshold $ \threshold 
        \geq 1 $.
        \Ensure An estimate of $ \edgecount( U ) $
        \State Uniformly and independently sample $ \Theta \fbrac{ \sqrt { \threshold } / \approxerror  } $ random edges from $ U $.
        \State Let $ \resamplesize = \Theta\fbrac{ \nicefrac{ \sqrt{\threshold} }{ \approxerror } }  $ be the number of samples drawn and $ \collisioncount $ is the number of collisions.
        \State \Return $\binom{\resamplesize}{2}/\collisioncount $.
    \end{algorithmic}
\end{algorithm}





\birth*

\begin{proof}
    Let $\sbrac{e_1, e_2, \ldots, e_{\resamplesize}}$ be the sampled edges. or $1 \le i < j \le \resamplesize$, define the indicator random variable $X_{ij}$ to be $1$ if $e_i = e_j$ and $0$ otherwise . Then, 
    $$\Pr[X_{ij}=1] = \frac{1}{ \edgecount( U ) } \text{ and } \E[X_{ij}] = \frac{1}{ \edgecount( U ) }$$
    Define the total number of collisions by $ X = \sum_{1 \le i < j \le \resamplesize} X_{ij} $. By the linearity of expectation
    \begin{align*}
    \E[X] = \binom{\resamplesize}{2} \frac{1}{ \edgecount( U ) }
    \end{align*}
    Since the random variables $ X_{ij} $ are not independent. We will bound the variance using
    \begin{align}\label{eq: variance equation}
        \Var(X) &= \sum_{i < j } \Var(X_{ij}) + 2\sum_{i < j, k < \ell} \mathrm{Cov}(X_{ij},X_{k\ell})
    \end{align}
    Let us bound the first term $ \Var(X_{ij}) $ in~\eqref{eq: variance equation}
    \begin{align*}
        \sum_{i < j } \Var(X_{ij}) \leq \sum_{ i < j } \E[X_{ij}^2] \leq \binom{\resamplesize}{2} \frac{1}{ \edgecount( U ) } =  \E[X]
    \end{align*}
    Next, consider the covariance terms in~\eqref{eq: variance equation}. If $X_{ij}$ and $X_{k\ell}$ correspond to disjoint pairs, then they are independent and the covariance is zero. If they overlap in exactly one index (share exactly one index), then
    \begin{align*}
        \mathrm{Cov}(X_{ij},X_{j\ell}) &= \E[X_{ij}X_{j\ell}] - \E[X_{ij}]\E[X_{j\ell}]\\  
        &= \Pr[X_{ij} = X_{j\ell} = 1] - \frac{1}{ \edgecount( U )^2}\\
        &= \Pr[e_i = e_j = e_\ell] - \frac{1}{ \edgecount( U )^2 } \\
        &= 0
    \end{align*}
So, the bound on the variance is $ \Var(X)  = \sum_{i < j } \Var(X_{ij}) \leq \E[X] $. Using Chebyshev’s Inequality (Lemma~\ref{lemma: chebyshev}), we have
\begin{align}
    \nonumber\Pr[\abs{X - \E[X]} > \frac{\approxerror}{3} \E[X]] &\leq \frac{9\Var(X)}{\approxerror^2 (\E[X])^2}\\
    \nonumber&\leq \frac{9}{\approxerror^2\E[X]} &\text{$\fbrac{\Var(X) \leq \E[X]}$}\\
    \nonumber&= \frac{ 18\edgecount( U ) }{ \approxerror^2\resamplesize(\resamplesize - 1) } &\text{$\fbrac{ \E[X] = \binom{\resamplesize}{2} \frac{ 1 }{ \edgecount( U ) } }$}\\
    &\leq \frac{1}{ 16 } &\text{($ \edgecount_U \leq \threshold $ and $\resamplesize = \Theta ( \sqrt{ \threshold } /{ \approxerror } )$)}\label{ineq: collision bound}
\end{align}
To complete the proof, we need to show with probability at least $ 15/ 16 $
\begin{align*}
    \fbrac{1 - \approxerror} \edgecount( U ) \leq \binom{\resamplesize}{2}  \frac{1}{X} \leq \edgecount( U ) \fbrac{1 + \approxerror}
\end{align*}
From~\eqref{ineq: collision bound}, with probability at least $ 15/ 16 $
\begin{align*}
    \fbrac{1 - \frac{\approxerror}{3}}\E\tbrac{X} \leq X \leq \fbrac{1 + \frac{\approxerror}{3}} \E\tbrac{X} 
\end{align*}
This implies
\begin{align*}
    \Rightarrow \frac{ \edgecount( U ) }{\fbrac{1 + \frac{\approxerror}{3}}}\leq \binom{\resamplesize}{2}  \frac{1}{X} \leq \frac{\edgecount( U ) }{\fbrac{1 - \frac{\approxerror}{3}}}
\end{align*}
For any $ \approxerror \in (0, 1) $
\begin{equation*}
    1 - \approxerror \leq \fbrac{1 + \frac{\approxerror}{3}}^{-1} \text{ and } \fbrac{1 - \frac{\approxerror}{3}}^{-1} \leq 1 + \approxerror
\end{equation*}
Since $X$ is the number of collisions. we conclude that
    \begin{align*}
        \fbrac{1 - \approxerror} \edgecount( U ) \leq \binom{\resamplesize}{2} \frac{1}{\collisioncount} \leq \edgecount( U ) \fbrac{1 + \approxerror}
    \end{align*}
which completes the proof.
\end{proof}

\end{document}